\documentclass[10pt, conference, letterpaper]{IEEEtran}
 
\usepackage{graphicx} 
\usepackage{amsmath,amssymb,amsfonts}
\usepackage{graphicx}
\usepackage{textcomp}
\usepackage{algorithm}
\usepackage[noend]{algpseudocode}
\usepackage{color,soul}
\usepackage{xcolor}
\usepackage{subfig}
\usepackage{url}
\usepackage{mathtools} 
\usepackage{amsfonts}
\usepackage{nicefrac}
\usepackage{bm}
\usepackage{tablefootnote}
\usepackage[table]{xcolor}
\usepackage{colortbl}
\usepackage{graphicx}
\usepackage{nicefrac}
\usepackage{enumitem}

\usepackage{amsthm}
\newtheorem{theorem}{Theorem}
\newtheorem{corollary}{Corollary}

\newtheorem{definition}{Definition}

\newtheorem{remark}{Remark}
\usepackage{subcaption}
\usepackage{amsmath}
\usepackage{graphicx}

\title{Prioritizing Microservice Requests via Retries}
\author{
\IEEEauthorblockN{
Ofek Lavi\IEEEauthorrefmark{1} \;\;
Jhonatan Tavori\IEEEauthorrefmark{4} \;\;
Hanoch Levy\IEEEauthorrefmark{1} \;\;
Anat Bremler-Barr\IEEEauthorrefmark{1} \;\;
}
\IEEEauthorblockA{\IEEEauthorrefmark{1}Tel Aviv University \quad \quad
\IEEEauthorrefmark{4}Columbia University }
}

\begin{document}

\maketitle

\begin{abstract}

Cloud services rapidly shift towards independently-managed microservices architectures. This independence may lead to temporary microservices miscoordination, resulting in overloads and rejections of valuable transactions. This weakness may further be exploited by attackers to launch overload attacks inflicting performance and cost damages.

Retry operations are a key machinery aiming at overcoming seldom rejections. However, during overload periods they: a) Fail protecting valuable transactions from rejections, and b) Trigger large-scale performance-degrading retry storms.

We propose simple modifications in retry operations, overcoming these deficiencies by transforming retries into a rejection prioritization mechanism, \textit{DiffTry} (Differential Retries). It is based on a simple, lightweight, distributed and stateless design and requires neither overload sensing nor microservices coordination.

We develop an analytical model capturing \textit{DiffTry}'s effectiveness for two services in tandem, then generalize it to networks of services. Our theoretical framework lets practitioners quantify, for a given network, the resulting prioritization and rejection rates, and yields formal guarantees that \textit{DiffTry} substantially reduces the impact of both incidental and adversarial orchestrated overloads. We validate these predictions in simulation and on a live Kubernetes deployment, where \textit{DiffTry} is a simple drop-in service-mesh plugin requiring no application changes; the measured rejection closely matches the analysis.

\end{abstract}

\section{Introduction}

Cloud service providers have largely converged on microservice architectures, decomposing monolithic applications into hundreds, often thousands, of loosely coupled services \cite{uber_microservices,netflix_microservices}. 
This paradigm 
allows each component 
(microservice) 
to scale elastically in response to demand,
implementing capacity scaling, retries, timeouts, and circuit breakers autonomously.

While these mechanisms achieve high  microservice independence and simplicity in design,   their coordination is challenging and their operations can get miscoordinated. During periods of  miscoordination, in particular, miscoordination of service capacity,  some microservices may become overloaded due to their service capacity not up-scaled to  the load generated by upstream services. Such overload periods may result in significant losses of transactions, due to their  blocking/rejection by the downstream overloaded microservice. 

The negative economic effects of such rejection periods are further compounded when they affect high-value (high-priority) transactions. These may include premium transactions (such as a customer with a loaded cart, or a preferred customer), mission critical transactions and others.

Two main overload-prevention methodologies are used in the literature and in practice: 
(i) local auto-scaling approaches (e.g., \cite{kubernetes_autoscaling}), and  
(ii) network-wide centralized approaches (e.g., TopFull \cite{topfull2024}, which offers network-wide source throttling).
Local auto-scaling reacts to overload by increasing resources at individual nodes. While it can respond quickly, it lacks a global view and therefore cannot address traffic miscoordination across services. Centralized approaches have such a global view, but require collecting network-wide information and are therefore slower to react. The resulting time gap can create temporary periods of miscoordination and overload, during which high-value transactions may still be rejected.
Malicious users can exploit this weakness by injecting spikes of cheap low-value transactions, overloading critical microservices and causing mass rejections of innocent high-value traffic, thereby inflicting economic damage on the service \cite{yoyoinfocom2017}.

In this paper, we propose \textit{DiffTry} as a mechanism for service differentiation and load control using retries.

\subsection{\textit{DiffTry}: Differential Retries}

\textbf{\textit{DiffTry}} (Differential Retries) is a simple and stateless mechanism for differential operation of retries across transaction classes, yielding differentiation (and prioritization) in their rejection rates. 
In a nutshell, \textit{DiffTry} intervenes right at the point where rejections occur, protecting high-priority transactions during overload by assigning class-specific retry limits.
The key idea is to assign a different number of retry attempt limits to different classes or tiers of requests. For example, a premium user may be allowed up to three attempts, while a "bypassing user"  may be limited to one. These limits can be configured per route, user tier or traffic class (e.g., with priority carried in an HTTP header), and enforced locally at the retry decision point (e.g., the sender API).

By embedding the prioritization logic directly into the retry process, \textit{DiffTry} eliminates the need for explicit coordination or tight coupling between services. It effectively transforms retries into a distributed rejection-prioritization control, enabling systems to both prioritize critical flows and suppress retry amplification in a local, coordination-free, scalable manner.

During attacks, this may have a particularly critical effect, as it protects highly valued transactions from large volumes of attacking requests (naturally classified as low-priority tasks).
Furthermore, without requiring explicit overload detection or sensor, the mechanism has an effect only under high-load or overload conditions and remains effectively silent otherwise.

To illustrate \textit{DiffTry} under overload, consider a front-end Service~A and a downstream Service~B (Fig.~\ref{fig:DiffTry-comparison}). Requests are labeled by priority, either \textit{High} (e.g., premium or latency-sensitive users) or \textit{Low} (e.g., a bypassing user or batch job), and Service~A retries them after timeouts or errors.

\begin{figure}[t]
    \centering
    \subfloat[
    Same number of retry attempts: rejections of both high- and low-priority.
    ]{
        \includegraphics[width=0.95\linewidth]{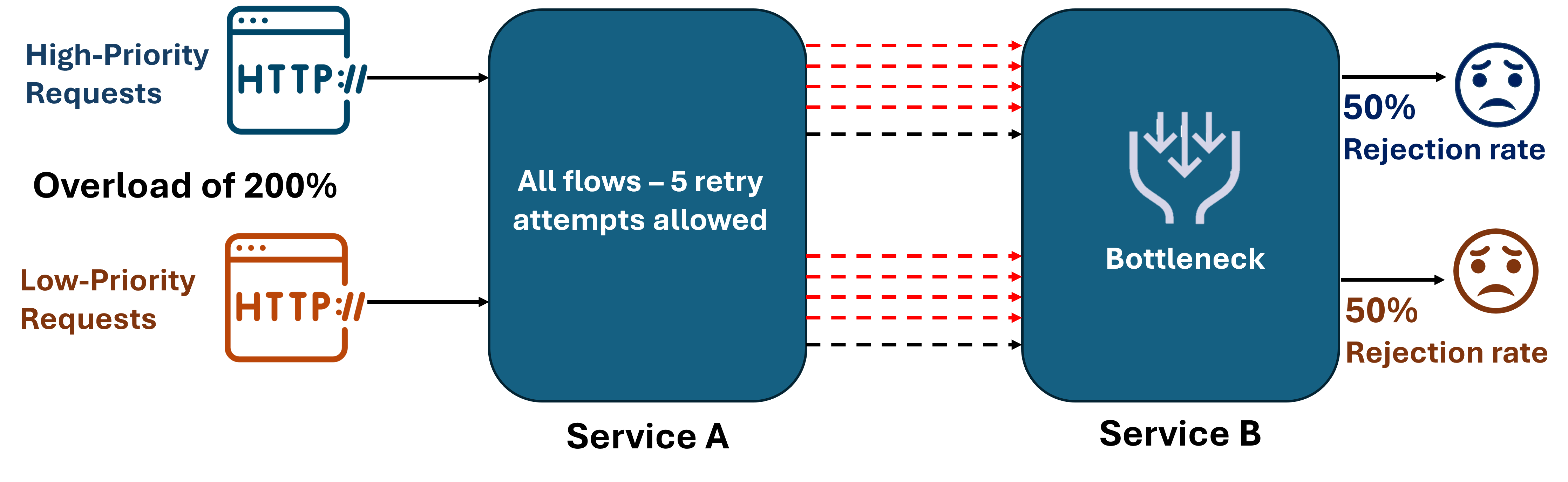}
        \label{fig:noDiffTry}
    }
    \hspace{0.9cm}
    \subfloat[ 
    \textit{DiffTry}: differentiate retry limits, High-priority enjoy low rejection rates.
    ]{
        \includegraphics[width=0.95\linewidth]{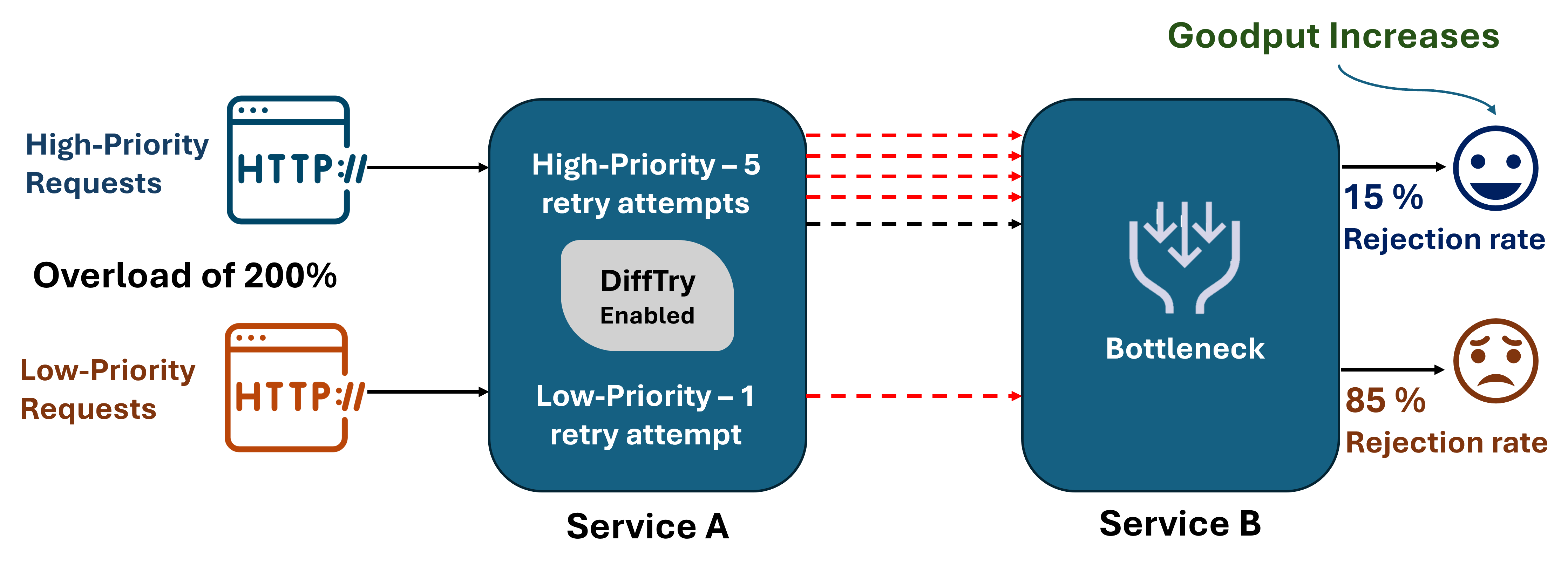}
        \label{fig:DiffTry}
    }
    \caption{\textit{DiffTry} impact on retry behavior under overload.}
    \label{fig:DiffTry-comparison}
\end{figure}

In classical retry mechanisms (a), when Service~B becomes overloaded due to a traffic spike or scaling lag, Service~A  retries failed requests, ignoring  priority. Thus, transactions are rejected and delayed regardless of value. Although overload is often driven by low-priority traffic, and can be amplified by low-priority attack traffic, the main economic damage comes from degraded high-priority requests.
In contrast, with \textit{DiffTry} (b), high-priority requests receive more retry attempts, increasing their success probability and reducing their rejection rate to around 15\%. This selective persistence significantly improves their \textit{goodput}: even when both classes contribute 50\% of the input load, high-priority transactions account for about 85\% of successful completions, corresponding to a throughput improvement of about 70\%. This tradeoff can be further enhanced by adding retries to high-priority requests.

\subsection{Main Contributions:  \textit{DiffTry} Analysis and Results}
\label{subsec:results}

Our paper focuses on establishing and analyzing the behavior of \textit{DiffTry}. This should serve designers in (i) Gaining  trust in this new machinery, (ii) Using the analysis to optimize the operations of \textit{DiffTry} and plan the system 
performance.

We start by presenting the \textit{DiffTry} mechanism, properties, and model in Sec. \ref{sec:mechanism}. Then, 
we analyze a sub-system in which the upstream Service A sends traffic to the downstream Service B (Sec. \ref{sec:analysis}). We first focus on the overload region, and using a fluid-model approximation we derive the per-stream rejections and throughput. 
We then consider the no-overload region, where rejection depends on the underlying queueing model. We show how queueing-specific blocking expressions can be plugged into our retry analysis, to derive the per-stream rejections; we demonstrate this on the M/M/1/m model.
Simulation results demonstrate the model accuracy in Sec. \ref{sec:evaluation}.  

Having established the analysis for a single overloaded node  we turn to deal with networks (Sec. \ref{sec:network}). 
This is motivated by real microservice systems, which often contain tens or hundreds of services, e.g., Uber~\cite{uber_microservices} and Netflix~\cite{netflix_microservices}.
There are two key motivating questions behind this analysis: 1) Whether the plausible behavior (low rejections for high-priority transactions)  derived  for a single node would hold under an arbitrary network. In particular, could  high-priority class  enjoy low rejections in one node, but  affect the traffic in a way that it will suffer high rejection rate  on  other (downstream) nodes. 2) 
Approximate the network rejections under \textit{DiffTry}.  

We adopt the well known open-network model \cite{kleinrock1975queueing1}, combined with flow classes/priorities,  to model the traffic in the network. We then extend it to model \textit{DiffTry}.
By incorporating the rejection equations we generate a set of equations whose solution 
(computed in a simple iterative approach) yields the rejection rate of each priority class in each node. This allows deriving the end-to-end rejection rates. 
We further prove that the iterative approach is always guaranteed to converge to a solution and derive network rejection probabilities.
We also investigate whether the solution is guaranteed to be unique (Sec. \ref{sec:uniqueness}). 
We derive a set of conditions under which uniqueness is guaranteed. In particular it holds for a wide set of applications where the transaction flow graph is represented by a DAG (examples include Online Boutique \cite{google_online_boutique} and WeChat \cite{DAGOR}).  For other cases, we prove uniqueness for some cases, but the general theoretical problem remains open in the context of this work; for those cases, we conduct an extensive set of experiments, all of which yield a unique solution.  

We establish monotonicity results (Sec.~\ref{sec:dominance}) showing that network-wide class-level success rates increase monotonically with retry limits. This establishes the effectiveness of DiffTry as well as allowing establishing performance lower bounds.

Lastly, in Sec.~\ref{sec:exp} we study \textit{DiffTry}'s network-level behavior
numerically and validate it on a live Kubernetes deployment.
There, \textit{DiffTry} is a \emph{drop-in service-mesh plugin} that needs no
application-code changes, and the measured rejection closely matches our
analysis, confirming that the local mechanism composes into network-wide
prioritization on a real system.
Concluding remarks are given in Sec. \ref{sec:conclusion}. 

\section{Related Work}\label{sec:related_work}
Traffic and overload control are challenging in microservice systems due to miscoordination among independently auto-scaled services. Prior work shows that uncoordinated scaling decisions can lead to instability and excessive cost under burst workloads \cite{luo2021characterizing,xie2024pbscaler}.
Several recent studies address these challenges using centralized control. For example, \cite{topfull2024} performs top-down overload control at the application ingress, \cite{wang2024autothrottle} uses an application-wide controller to set global performance and resource targets, and \cite{lim2024opportunities} relies on centralized request routing and traffic engineering across services and replicas.

In cloud systems, retry mechanisms are tied to service communication protocols, including HTTP \cite{yu2019survey} and RPC \cite{nikita2021}. Major cloud providers integrate automated retry logic directly into client and service libraries \cite{mathew2014overview, azure_retries}, and service mesh platforms provide built-in retry handling at the infrastructure layer \cite{istio_traffic_managment, linkerd_traffic_managment, saleh2022empirical, ashok2021leveraging}; all with a constant limit per route.

Traditional priority mechanisms (e.g., \cite{kleinrock1975queueing1}) are offered by cloud providers  \cite{googlecloud_request_priorities,lunar_priority_queue_api}, but their capabilities differ from \textit{DiffTry}'s. They attach the priority to the \emph{queue discipline}, whereas \textit{DiffTry}
attaches it to the \emph{retry budget}. They were designed to prioritize  transaction processing and are applied at the service point (of Service~B, in Fig.~\ref{fig:DiffTry-comparison}), where transactions exit the queue and enter service, as opposed to \textit{DiffTry}, which is applied before queue entrance, at the {\it rejection point}. As such, unless coordinated with other mechanisms, they would have difficulties handling rejections. Also, in no-overload situations, they may unnecessarily starve low-priority transactions and inflict long delays. An alternative approach, consisting of prioritization of rejections at the entrance to Service~B, would be challenging to implement as well, since it would require coordination with the buffer management of Service~B.

Recently, major cloud providers and service-mesh platforms have warned that retries can trigger retry storms and lead to performance problems and security threats \cite{microsoft_retry_storm, googleselfinf, tavori2023cost}. 
However, systematic analysis of retries under load remains limited, with only a small number of studies examining its impact on throughput and latency in microservice systems~\cite{mendonca2020model, saleh2022empirical}. 
DAGOR~\cite{DAGOR} also protects high-priority requests, but sheds at admission, requiring each server to sense overload and propagate its admission level upstream; RetryGuard~\cite{retryguard} prevents retry storms, but without differentiating classes. \textit{DiffTry} protects important jobs through the retry budget, needs neither overload detection nor coordination, and so acts immediately.

\section{\textit{DiffTry}: Mechanism Operation and Model} \label{sec:mechanism}

As illustrated in Fig. \ref{fig:DiffTry-comparison}, the key idea of \textit{DiffTry} is to assign different retry-attempt limits to requests from different \textit{priority} classes. 
By doing so, \textit{DiffTry} ensures that high-priority requests receive preferential treatment under overload, while, at the same time, suppresses retry storms. 
\textit{DiffTry} is especially effective during miscoordination-induced overload periods and helps reduce the economic costs associated with it. It addresses three key problems in existing microservices applications:
\begin{enumerate}
    \item \textbf{Insensitivity to priority:} High-priority requests  are rejected  as much as  low-priority ones.
    \item \textbf{Retry amplification:} Excessive retries inflate load and waste capacity and CPU.
    \item \textbf{Over-scaling:} Retried low-priority traffic triggers unnecessary scaling, increasing cost.
\end{enumerate}

This simple differentiation allows \textit{DiffTry} to operate effectively across a range of conditions and it possesses the following properties: (1) During \textbf{overload times}, \textit{DiffTry} shifts the mass of rejections caused by overload periods toward low-priority (non-latency-sensitive or non-revenue-critical) tasks, thereby protecting critical traffic. (2) During \textbf{attack times}, it significantly reduces the economic damages of an attack by preferentially rejecting low-priority requests that would otherwise overwhelm the system. (3) During \textbf{no-overload times}, the mechanism remains non-intrusive and does not interfere with normal traffic, while automatically becoming effective when overload conditions arise. (4) \textit{DiffTry} is designed for \textbf{distributed operation}, without centralized coordination or explicit overload detection, and relying only on local retry decisions. (5) Through \textbf{retry suppression}, \textit{DiffTry} limits excessive retries from low-priority flows, thereby reducing the risk of retry storms and stabilizing system behavior.

We formally model \textit{DiffTry} as follows: 
Service A receives requests from $P$ distinct sources where the arrival rate of class  $p$ is $\lambda_p$.    
It processes the incoming requests and forwards them to Service~B (as in Fig.  \ref{fig:DiffTry-comparison}), whose processing rate is  $\mu$ requests/sec. If Service B is unable to accept a request due to overload, Service A retains the request and retries it later.

    Each flow $p$ is assigned a retry budget $k_p$. This means that  for every $p \in \{1, \ldots, P\}$ flow $p$ may retry a failed request up to $k_p$ times.
Let $\lambda_p^i$ denote the rate of requests from stream $p$ currently in their $i$-th retry attempt, for $i = 1, \dots, k_p$ and $p=1,\dots,P$. A request of type $p$ is dropped permanently if it reaches its maximum retry budget ($k_p$) without being successfully served. During the first $k_p$ attempts, it is eligible for retries, governed by the retry logic of Service A.

Using this model, we next study \textit{DiffTry}'s operation using a two-tier microservice system consisting of two services. 
In Sec. \ref{sec:network}, we extend this model to account for an entire network.

\section{Analysis of \textit{DiffTry} in a single service}\label{sec:analysis}

In this section we analyze \textit{DiffTry} in a single service, derive its performance metrics, and show its ability to enforce prioritization and suppress retry amplification under load.

\subsection{\textit{DiffTry} 
Dynamics Under Overload (Fluid Model)}
\label{subsec:retry-analysis}

We first quantify the steady-state load and the rejection probability of an arbitrary request attempt under overload, when \textit{DiffTry} is active.
We use a \textit{fluid approximation}, which allows us to 
capture
the retry dynamics under overload (when rejection is mostly due to the bounded capacity and not due to finite buffers limitations), while abstracting away queueing details.\footnote{A fluid model approximates discrete request arrivals and services as continuous flows, capturing average system behavior through rates rather than individual jobs. Such models are commonly used to analyze large-scale or heavily loaded systems, where stochastic variability averages out \cite{lenin2000fluid, gamarnik2016delay}.} 
The main advantage of the fluid model is twofold: It enables us to focus on the retry dynamics, and
it ensures that the analysis is independent of the specific queueing discipline, covering a wide class of queueing models. We validate the accuracy of its results through simulations in Sec. \ref{sec:evaluation}.

We assume the system is operating under persistent overload, i.e., $\sum_{p=1}^P{\lambda_p} > \mu$. In this setting, the system reaches a steady state where requests are continuously dropped and retried, but the retry traffic stabilizes due to capped retry budgets.\footnote{Unlike some of the classical models, where steady state exists only when $\rho  < 1$, here the retry cap ensures a fixed bound on retry-induced load, similarly to finite-buffer queueing models \cite{kleinrock1975queueing1}.} The total request load $\Lambda$ on Service B, including both original requests and retries, is given by:
\begin{equation}\label{eq:lambdasum}
    \Lambda = \sum_{p=1}^{P} \left( \lambda_p + \sum_{i=1}^{k_p} \lambda_p^i \right) = \sum_{p=1}^P{\lambda_p} + \sum_{p=1}^P{\sum_{i=1}^{k_p} \lambda_p^i};
\end{equation}
We define the \emph{utilization} (including retries) as $\rho={\Lambda}/{\mu}$ and the \emph{pure utilization} (excluding retries) as $\rho_{\text{pure}}{=}\sum_{p=1}^P{\lambda_p}/\mu$.

Let $s$, $\bar{s} {=} 1 {-} s$  respectively denote the success and rejection probabilities of a single request being served by Service~B. 
They implicitly depend on the total presented load (original and retried), which affects the queueing behavior and acceptance rate.
Our analysis is based on a simplifying approximation, \textit{an independence assumption}, which will  be validated via simulation (Sec. \ref{subsec:validation_assumptions}) (and for extreme cases, relaxed on Sec. \ref{subsec:relax}): We assume that the rejection probability $\bar s$ is independent of its stage $i$; this is likely to hold if the delay between retry attempts is large enough relative to the service time, which is the case in cloud services \cite{scaling_exp_backoff,aws_retry_backoff}.
Under this assumption, retry flows follow a geometric progression, $ \lambda_p^i = \bar{s}^i \cdot \lambda_p$ for $i {=} 1{,} \dots, k_p$. Thus, the overall arrival rate is:  
\begin{equation}
    \label{eq:biglambda}
    \Lambda =  \sum_{p=1}^P{\sum_{i=0}^{k_p} \lambda_p \cdot \bar{s}^i} = \sum_{p=1}^P{\lambda_p \cdot \left( \frac{1 - \bar{s}^{k_p+1}}{1 - \bar{s}} \right)}
\end{equation}  
Since Service B serves at rate $\mu$ while busy, then in steady state we have $\Lambda\cdot(1-\bar{s})=\mu\cdot(1-p_0(\rho))$, where $p_0(\rho)$ is the probability that Service B is idle. This gives:
\begin{equation}
\label{eq:s}
    \bar{s} \;=\; \Bigl(1-\frac{1}{\rho}\Bigr) + \frac{p_0(\rho)}{\rho} \;\ge\; 1-\frac{1}{\rho},
    \qquad \text{for $\rho>1$} .
\end{equation}
The \textit{fluid approximation} takes eq. \eqref{eq:s} with equality, i.e.\ $p_0\!=\!0$, neglecting the idle capacity. As $p_0(\rho)$ is non-increasing in $\rho$ and vanishes as the load grows, it is accurate in the overload regime \textit{DiffTry} targets.
Thus we have two equations in two unknowns ($\Lambda, \bar s$) which can be solved numerically.
Practically, 
they translate into a single equation in one unknown $\bar s$:
\begin{equation}
\label{eq:s-equation}
    \mu = {\sum_{p=1}^P \lambda_p \left( 1 - \bar{s}^{k_p+1} 
\right)}.
\end{equation}
Eq.~\eqref{eq:s-equation} clearly admits a unique solution $\bar {s}^* {\in} (0,1)$; 
it quantifies how the retry budgets $(k_1, ..., k_P)$, arrival rates $(\lambda_1, ..., \lambda_P)$ and capacity $\mu$ determine the system’s rejection rate.

\subsection{A More Detailed Model for No Overload}
\label{sec:no-overload-new}

For completeness, and for the network analysis in Sec.~\ref{sec:network}, we next consider the no-overload regime. 
    For  $\rho {\le} 1$, the fluid approximation would simply yield $\bar s =0$. Combining it with the analysis above yields a monotone continuous approximation, valid for the full range of operations. 
    To obtain
    a more precise approximation, one can replace Eq. (\ref{eq:s}) with  a (system dependent) queueing model to provide a continuous expression 
    for the rejection probability as a function of $\rho$: $\bar{s}=\bar{s}(\rho)$.

For example, for an $M/M/1/m$ queue,
$\bar{s}(\rho)=
\frac{(1-\rho)\rho^m}
{1-\rho^{m+1}}$, which is the well-known blocking probability of the model~\cite{kleinrock1975queueing1}.
Given $\{ \lambda_p \}, \mu$ the corresponding rejection probability $\bar{s}^*$ can be calculated by plugging this equation into Eq. (\ref{eq:biglambda}) and solve for $\rho$.

In order to guarantee a continuous approximation over the entire operating range, we combine the queueing-based model for $\rho\le1$ with the fluid approximation for $\rho>1$. If $p_0$ can be calculated given the queueing model, one can plug it to eq. \eqref{eq:s}. Otherwise, since $\bar{s}$ is nondecreasing in $\rho$, then $\bar{s}(1)$ is a lower bound for $\rho>1$, so combine our bounds to define:
\begin{equation}
\hat{s}(\rho)=
\begin{cases}
\bar{s}(\rho), & \rho\le1,\\[1ex]
\max\!\left(1-\nicefrac{1}{\rho},\,\bar{s}(1)\right), & \rho>1,
\end{cases}
\label{eq:combined_sbar}
\end{equation}
where $\bar{s}(1)=\lim_{\rho\rightarrow1^-}\bar{s}(\rho)$ whenever necessary. The resulting function is continuous and monotone; thus,  replacing $\bar{s}(\rho)$ by $\hat{s}(\rho)$ in Eq.~\eqref{eq:biglambda} yields a unified fixed-point equation.

\subsection{Per-Stream Throughput and Rejection Rates}
\label{subsec:throuput}
Once $\bar{s}$ is solved numerically, we can now compute the contribution of each stream to the total load and the 
effective throughput.
Let $\Lambda_p$ denote the total request rates from stream p, including both original and retry attempts. Using the geometric sum derived earlier, we get for every $p\in[P]$:
$\Lambda_p = \lambda_p \cdot \sum_{i=0}^{k_p} \bar{s}^i = \lambda_p \cdot \left( \nicefrac{1 - \bar{s}^{k_p+1}}{1 - \bar{s}} \right)$.
Each request attempt has a success probability \( s = 1 - \bar{s} \). Thus, the \textit{expected} successful completions rate for class $p$ is:
$E_p = \Lambda_p \cdot s = \lambda_p \cdot ({1 - \bar{s}^{k_p+1}})$.

We now define the {\it request success probability} of class $p$: $S_p$, as the fraction of class-$p$ requests that are successfully served. Likewise, let $\bar{S}_p$ denote the corresponding {\it request reject probability}:
$S_p=\frac{E_p}{\lambda_p}=1-\bar{s}^{k_p+1}$
and
$\bar{S}_p=\bar{s}^{k_p+1}$.

\begin{figure*}[t]
    \centering

    \begin{minipage}[t]{0.48\textwidth}
        \centering
        
        \begin{minipage}[t]{0.48\linewidth}
            \centering
            \includegraphics[width=\linewidth]
            {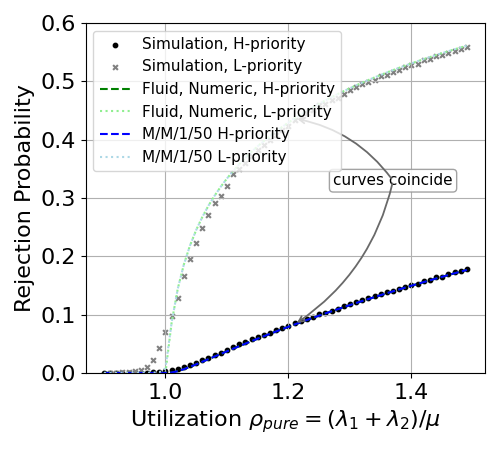}
            {\small (a) $(k_1, k_2)=(5,1)$}
        \end{minipage}
        \hfill
        \begin{minipage}[t]{0.48\linewidth}
            \centering
            \includegraphics[width=\linewidth]
            {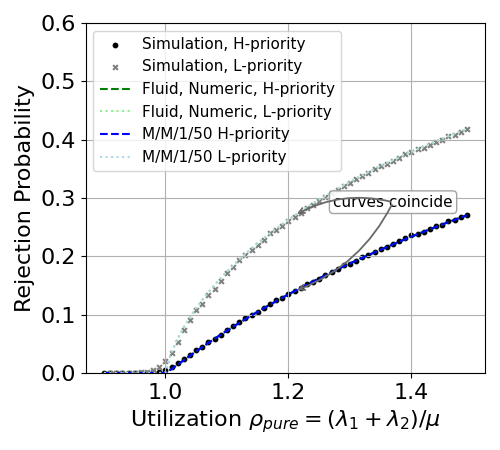}
            \\
            {\small (b) $(k_1, k_2)=(5,3)$}
        \end{minipage}

        \caption{Simulation VS Analytic: Rejection probability of original requests with \textit{DiffTry} in a single $M/M/1/50$ service.
        }
        \label{fig:validation_retry_budgets}
    \end{minipage}
    \hfill
    \begin{minipage}[t]{0.48\textwidth}
        \centering
        
        \begin{minipage}[t]{0.48\linewidth}
            \centering
            \includegraphics[width=\linewidth]
            {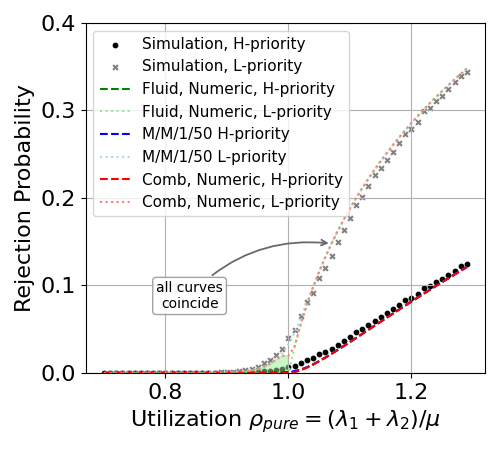}
            {\small (a) Typical (cloud-like) buffer size ($m=50$) }
            \label{fig:validation_mixed_a}
        \end{minipage}
        \hfill
        \begin{minipage}[t]{0.48\linewidth}
            \centering
            \includegraphics[width=\linewidth]
            {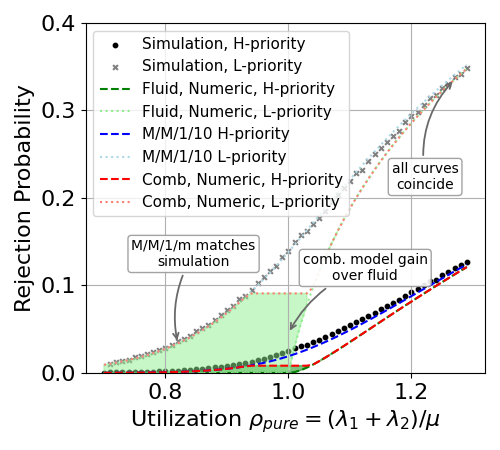}
            {\small (b) Atypical (corner case) buffer size ($m=10$)}
            \label{fig:validation_mixed_b}
        \end{minipage}

        \caption{Model behavior near the critical load $\rho\approx 1$.
        }
        \label{fig:validation_mixed}
    \end{minipage}

\end{figure*}

This captures the tail behavior of each flow. Lower $k_p$ values sharply reduce the retry-induced load and limit the number of requests that persistently stress the system under overload.

\subsection{Relaxing the Independence Assumption for Highly Variable Service Times}
\label{subsec:relax}
The model above assumes retries fail independently, which
holds when the timeout exceeds the service time. Under highly variable service times a
retry issued soon after a rejection may still meet the \emph{same} long job,
correlating failures. Conditioning on the residual $R$ of that job, retries before it fail deterministically, while retries afterwards have a fresh probability $\bar{s}$ since the queue renews. Writing $\bar{S}_n$ for rejection on the initial
attempt plus $n$ retries and
conditioning on
$R$:
\[
\mathbb{P}(\bar{S}_k) = \sum_{i=0}^{k} \mathbb{P}(\bar{S}_k \mid t_i < R \le t_{i+1}) \cdot \mathbb{P}(t_i < R \le t_{i+1}),
\]
Which can be calculated recursively given $\mathbb{P}(t_i < R \le t_{i+1})$:
\[
\mathbb{P}(\bar{S}_k) = \bar{s} \cdot \left[ \mathbb{P}(R > t_k) + \sum_{i=0}^{k-1} \mathbb{P}(\bar{S}_{k-i-1}) \cdot \mathbb{P}(t_i < R \le t_{i+1}) \right].
\]
For constant timeouts with memoryless service times (like an $M/M/1$ server), each interval has the same chance $\alpha=\mathbb{P}(R<t)$, collapsing the equation to
\[
\mathbb{P}(\bar{S}_k)=\mathbb{P}(\bar{S}_{k-1})[(1-\alpha)+\alpha\bar{s}] = 
\bar{s} \cdot \Big( (1-\alpha) + \alpha \bar{s} \Big)^k .
\]

When $\alpha{=}1$, the blocking job completes before the next retry attempt, giving the independence assumption, $\bar{S}_p{=}\bar{s}^{k_p+1}$. Smaller $\alpha$ values represent stronger correlation, as retries are more likely to observe the same blocking job.

\section{Examination of \textit{DiffTry} in a Single-Service}\label{sec:evaluation}

We next examine \textit{DiffTry} in  a single-service setting, with three objectives:  (i) Evaluate the accuracy of the independence assumption and the fluid approximation (see Sec. \ref{subsec:retry-analysis}). (ii) Examine how effective are the retry-limits in prioritizing rejections across classes at overload situations. (iii) Examine \textit{DiffTry} in no-overload situations.
We use two examination machineries: a discrete-event simulation of an \textit{M/M/1/m} system with retries, and numerical evaluation of both the \textit{M/M/1/m} model and the fluid model. We validate the simulator against the exact \textit{M/M/1/m} numerical evaluation,
and assess the fluid approximation by comparing it to the \textit{M/M/1/m} evaluation. The code of the simulations is available in \cite{difftry_code}.

\subsection{Examination in typical cases: Accuracy and Effectiveness}
\label{subsec:validation_assumptions}
We aim at evaluating the model in typical situations. We consider a system as in Sec. \ref{sec:analysis} and in Fig. \ref{fig:DiffTry-comparison}, and assume that the buffer of Service B has room for 50 tasks; we expect that in typical application this number will be in the order of 10s or 100s.  We further assume that the retry timeout lasts in the order of several service times (of Service B).

Fig.~\ref{fig:validation_retry_budgets} validates the model across retry budgets: we fix $k_1$, vary $k_2<k_1$, and sweep the offered load via $\lambda_2$ (with $\lambda_1$ fixed).
The simulation closely matches both the $M/M/1/m$ and fluid-model derivations, supporting our approximations.

Beyond validating the approximations, Fig.~\ref{fig:validation_retry_budgets} shows two central effects.
First, increasing the low-priority retry budget  increases the volume of retry traffic and aggregate load, demonstrating \emph{retry-amplification} (storm).
Second, \textit{DiffTry} enforces \emph{prioritization}: across traffic mixes, the class with the larger retry budget has substantially lower ultimate rejection probability, while the lower-priority class absorbs most losses.

\subsection{
Examination of the model in extreme cases}
\label{subsec:validation_mixed}

By eq. \eqref{eq:s}, the fluid approximation errs by $p_0/\rho$ above capacity, and below it by the whole rejection, since it assumes $\bar{s}\approx0$ for $\rho\le1$. For \textit{M/M/1/m} both are largest at $\rho=1$; the error is therefore at most $1/(m+1)$ everywhere, and shrinks geometrically in $m$ away from criticality.
Fig.~\ref{fig:validation_mixed} evaluates this failure mode by comparing two buffer sizes.
In the small-buffer setting ($m=10$, Fig. \ref{fig:validation_mixed}(b)), $1/(m+1)$ is not small, so rejection is already non-negligible below overload, causing the fluid model to deviate near $\rho=1$.
The combined approximation addresses this by using the exact \textit{M/M/1/m} rejection for $\rho\le 1$ and transitioning to the overload lower bound for $\rho>1$ (Eq.~\eqref{eq:combined_sbar}), yielding a closer match where the fluid approximation is least reliable.
We stress that \textit{DiffTry} primarily targets cloud microservices, where $m$ is in the tens or hundreds and $1/(m+1)$ is correspondingly negligible, making the fluid approximation accurate (Fig.~\ref{fig:validation_mixed}(a)).
Nevertheless, below-overload rejection may arise with small buffers or effective capacity reductions, e.g., transient contention or scheduling effects. In such cases, the mixed approximation improves robustness while preserving tractability (Fig.~\ref{fig:validation_mixed}(b)).

To probe the residual correction of Sec.~\ref{subsec:relax}, we examine the regime it
is designed for: a small retry timeout, so a retry can strike the \emph{same}
still-blocking job. Fig.~\ref{fig:validation_residual} reports the per-class rejection
in this short-timeout regime under (a)~low and (b)~high service-time variance; we induce
high variance with a two-phase (hyper-exponential) service in which a job is 'long'
with probability $q=0.01$ (mean service $20\times$ larger) and 'short' otherwise. With low
variance the correction is small (it stems from the short timeout alone) and the residual
model matches the simulation. With high variance (or low retry timeouts) the independence assumption breaks
(correlated failures), yet the residual model still tracks the
simulated rejection. We also tested many variances and smaller timeouts; all gave similar results. Hence even in this adversarial corner, 
our residual model accurately predicts the rejection rate.

\begin{figure}[t]
    \centering
    \begin{minipage}[t]{0.48\linewidth}
        \centering
        \includegraphics[width=\linewidth]{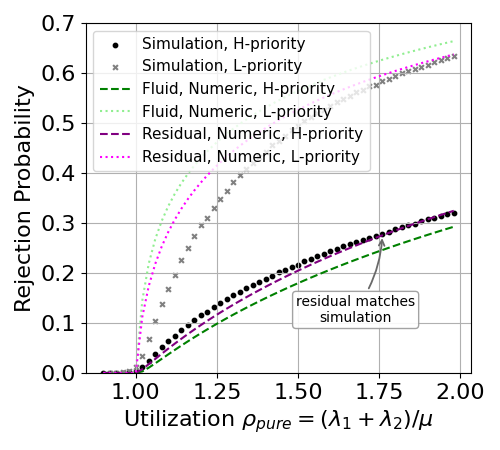}
        {(a) low variance ($\mu_{\text{high}}=\mu_{\text{low}}$)}
    \end{minipage}
    \hfill  
    \begin{minipage}[t]{0.48\linewidth}
        \centering
        \includegraphics[width=\linewidth]{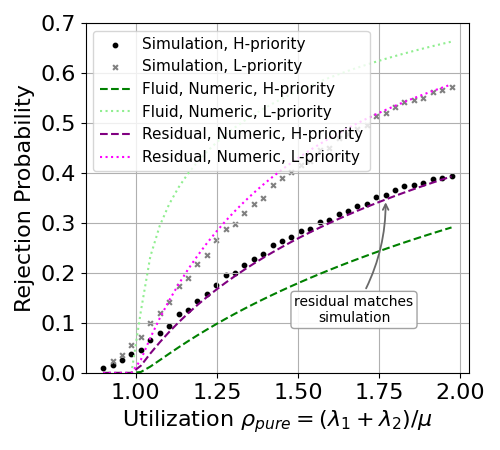}
        {(b) high variance ($\mu_{\text{high}}=20\cdot\mu_{\text{low}})$
        }
    \end{minipage}
    \caption{low retry timeouts ($R=t$), high/low variance ($q=\frac{1}{100}$)}
    \label{fig:validation_residual}
\end{figure}

\section {Analysis of \textit{DiffTry} on a General Network}\label{sec:network}

The analysis provided above dealt with the application of \textit{DiffTry} in  two interacting services. A general system may consist of a network of tens/hundreds of services (e.g., TopFull \cite{topfull2024}).
While the application of \textit{DiffTry} is local to these two services its effects may propagate through the network. 
Therefore an analysis of \textit{DiffTry} over a general network is needed. 

To conduct  the analysis we adopt the standard open network flow model (as used in Jackson networks \cite{kelly1979reversibility,kleinrock1975queueing1}, and in more general open network models, including priorities \cite {baskett1975bcmp}). We then generalize that model to accommodate the   \textit{DiffTry} equations; their non-linearity makes the analysis non-trivial. 

\subsection {The Standard Network Flow Model}

Under the standard open network flow models the network of services consists of $N$ services (nodes). External traffic of transactions (jobs)  enters node $i$  at rate $\gamma_i$. 
In Jackson networks it is further assumed that $\gamma_i$ is Poisson. Since our analysis deals only with rates 
we do not require this assumption. 

We further assume that the jobs are divided into classes $p{=}1,  ..., P$; let $\gamma_i^{(p)} $ denote  the external arrival rate  of class $p$ to node $i$.
After completion of job at node $i$, a fraction $f_{i,k}^{(p)}$ is directed to service $j{=} 1,..,N$, where $\sum_{j=1}^N f_{i,j}^{(p)} \leq 1$; a fraction $1{-}\sum_{j=1}^N f_{i,j}^{(p)}$ terminates and leaves the system.
The total rate of class $p$ jobs entering node $i$ is denoted by $\lambda_i^{(p)}$ and is given by the equation set (known as the {\it traffic equations}): 
\begin {equation}
\label{eq:linear}
\lambda_i^{(p)}= 
\gamma_i^{(p)} + 
\sum_{j=1}^N 
\lambda_j^{(p)} f_{j,i}^{(p)} \qquad i=1,...,N, p=1, ..., P.
\end {equation} 
It may be convenient to use vector notation terminology where $\lambda$ and $\gamma$ are vectors of size $NP$ holding the $\lambda_i^{(p)}$ and $\gamma_i^{(p)}$ variables, respectively, and $F$, called the {\it routing matrix},  is a matrix whose dimension is $NP \times NP$. Then Eq.  \ref{eq:linear} is expressed as $\lambda = \gamma + \lambda F$.
Deriving the set of arrival rates $\lambda_i^{(p)}$ for all $i$ and $p$ requires solving this linear equation set. This leads to the solution $\lambda = \gamma (I-F)^{-1}$.
It is well known that  this equation set can be solved   by iterating it on the vector $\lambda$:
\begin {equation}
\label{eq:iter-solution}
\lambda =
\gamma + \lambda F = 
\gamma + \gamma F + \lambda F^2 =
\dots = 
\gamma \sum_{i=0}^\infty F^i.
\end {equation}
A standard sufficient condition for the convergence of this sum (and, subsequently, for Eq. \eqref{eq:linear} to have a solution) in open queueing networks is that the spectral radius of the routing matrix $F$ (which is also assumed to be substochastic) satisfies $\rho(F) < 1$, which ensures that the fundamental matrix $(I-F)^{-1}$ exists and the traffic equations admit a unique nonnegative solution~\cite{kelly1979reversibility}.
\footnote{
Having a solution to the equations does not imply  stability of the system queues, which depends on their service rates.
}

\subsection {Generalizing the Standard Model to Accommodate \textit{DiffTry}}

To extend the model, we distinguish between the roles of $\lambda_i^{(p)}$ on the left- and right-hand sides of Eq.~\eqref{eq:linear}. 
Let $\lambda_i^{I;(p)}$ and $\lambda_i^{O;(p)}$ denote the input and output rates of  service $i$, respectively. 
Eq.~\eqref{eq:linear} becomes:
\begin{equation}
\lambda_i^{I;(p)}= 
\gamma_i^{(p)} +
\sum_{j=1}^N 
\lambda_j^{O;(p)} f_{j,i}^{(p)}.
\label{eq:linear-gen}
\end {equation} 

Now, for each service $i$ the differential retry mechanism will enter into effect as in Sec.
\ref{sec:analysis}. 
We denote the limit on the number of retries for node $1\leq i\leq N$ and for class $1\leq p \leq P$ as $k_i^{(p)}$. 
That is, differential rejections (resulting from differential retries) will operate on the set of $P$ priorities. 
As described in Sec.~\ref{sec:analysis}, $\bar s_i$ can be derived from the input parameters of this service: $(\mu_i, \lambda_i^{I;}, k_i^{;})$
where we use vector notations:$\lambda_i^{I;} := (\lambda_i^{I;(1)}, \dots, \lambda_i^{I;(P)})$ (and similarly$\lambda_i^{O;}) $, $k_i^{;} := (k_i^{(1)}, ..., k_i^{(P)})$.
Consequently, the rate at which class $p$ is processed at microservice $i$
is simply given by 
\begin {equation}
\label{eq:output}
\bar s_i = 
\bar s_i  (\mu_i, \lambda_i^{I;}, k_i^{;}), \quad
\lambda_i^{O;(p)}=\lambda_i^{I;(p)} (1- {\bar s_i} ^{k_i^{(p)} +1})
\end {equation} 
Eqs. (\ref{eq:linear-gen}) - (\ref {eq:output})
form an equations set in $2NP $ unknowns ($\lambda_i^{O;(p)}$, 
$\lambda_i^{I;(p)}$), whose solution will provide the resulting transaction rates in our system. This equation set is non-linear
and thus its solution, which we establish next, is not immediate.  

\subsection{Solving for the transaction rates for DAG Networks}
\label{subsec:solve_dag}

In a DAG, the system of equations~\eqref{eq:linear-gen} and~\eqref{eq:output} can be solved via a single forward pass along any topological order. When processing node $i$, the output rates $\lambda_j^{O;}$ of all predecessors $(j \to i) \in E$ are already known, allowing for the direct evaluation of the input vector $\lambda_i^{I;}$ via~\eqref{eq:linear-gen}. Using $\lambda_i^{I;}$ alongside the local parameters $\mu_i$ and $k_i^{;}$, we subsequently compute $\bar{s}_i$ and the corresponding output rates via~\eqref{eq:output}. This acyclic dependency ensures the procedure yields a solution and terminates in a single pass. Furthermore, by induction, local uniqueness at each node (Sec. \ref{sec:analysis}) guarantees uniqueness. Note that it is independent of the topological order chosen.

\subsection{Solving for the Transaction Rates in General Networks}
\label{sec:solve_trans}

The  methods of Sec. \ref{subsec:solve_dag} might not work for graphs containing cycles. Therefore, in this subsection we provide a general solution:  We propose  to use an iterative method on the $\lambda$ unknown variables, similarly to the standard iteration used on standard open queueing networks, and given in Eq. \eqref{eq:iter-solution}. 
To overcome the non-linearity, we propose to solve, on each iteration, for $\bar  s_i$ as given in Sec. \ref{sec:analysis}. This solution is performed in each iteration since $\bar s_i$ depends on the values of the $\lambda$ vector, which change over the iterations. For the sake of the analysis we will denote the iteration number by $t = 1, 2, ...$ and the corresponding $\bar s_i$ and $\lambda$ as  $\bar  s_i (t)$ and $\lambda(t)$, respectively. 

We next establish that this process converges by comparing it to the standard equation set given by Eq. \eqref{eq:linear}\footnote{In fact, this is a special case of our system,
whereby there are no rejections
},
showing that if eq. (\ref{eq:linear}) is solvable (converges)  then Equation set (\ref{eq:linear-gen}) --
(\ref {eq:output}) is also solvable (converges) via an equivalent iterative procedure. 

\begin {theorem}
\label {thm:convergence}
Under the conditions in which eq. (\ref{eq:linear}) has a solution (e.g. via an iterative procedure 
given by Eq. \eqref{eq:iter-solution}) then Equation set (\ref{eq:linear-gen}) -- (\ref {eq:output}) has also a solution. 
\end {theorem}

\begin{proof}[\textbf{Outline of Proof}]
The proof is carried out  using  an iterative solution and establishing its convergence.
Transforming eqs. \eqref{eq:linear-gen} -- \eqref{eq:output}
into iteration-index ($t$) dependent, and substituting them allows
re-writing them in the format of eq.~\eqref{eq:linear}: 
\begin {equation}
\lambda_i^{I;(p)} (t+1)= 
\gamma_i^{(p)} +
\sum_{j=1}^N 
\lambda_j^{I;(p)} (t)  \tilde f_{j,i}^{(p)}(\lambda_j^{I;}(t)),
\end {equation}
where
$\tilde f_{j,i}^{(p)}(\lambda_j^{I;}(t))
{:=}  S_j^{(p)} (t)  f_{j,i}^{(p)}$,
and $S_j^{(p)}(t) {:=} 1- \bar s_j^{k_j^{(p)} +1} (t)$ is 
the eventual success probability.
In vector form this gives: 
\begin {equation}
\lambda(t+1) = \gamma + \lambda (t) \tilde F (\lambda(t)).
\end {equation} 
Where $\tilde F(\lambda(t))$ is a matrix of size $NP\times NP$
expanded with zeros to get a $NP\times NP$ matrix.
The key point of the proof is that 
$S_j^{(p)} (t) $ is a probability, $0 \leq S_j^{(p)} (t)  \leq 1 $ and therefore  $\tilde F (t)$ obeys $0 \leq \tilde F (t) \leq F$, which leads to the proof.
\end{proof}

\subsection{Uniqueness of the Solution for cyclic graphs}
\label{sec:uniqueness}
In the previous subsection we showed that the network equations converge, and accordingly the  rejection probability of each edge can be derived.
While uniqueness for DAG topologies was proved,
for general networks \emph{with cycles}, uniqueness is non-trivial. The standard traffic
equations~\eqref{eq:linear} are linear and admit a unique solution whenever $\rho(F)<1$;
However, 
the DiffTry equations~\eqref{eq:linear-gen}--\eqref{eq:output} are \emph{nonlinear}—the
effective routing weights $\widetilde{f}^{(p)}_{i,j}=S^{(p)}_i(\lambda_i)f^{(p)}_{i,j}$
depend on the solution $\lambda$ itself through the shared per-node rejection probability
$\bar{s}_i$. Note that when all retry budgets are equal, this coupling vanishes and the system
reduces to a single-priority linear network; however, differentiated budgets break this structure,
and uniqueness requires a separate analysis.

We establish uniqueness by first proving it for two important cases (Thm. \ref{thm:uniqueness} next), and then providing support through numerical validations of the general case.

\subsubsection{\textbf{uniqueness proof}}
Theorem~\ref{thm:uniqueness} states the results.
\begin{theorem}[Sufficient conditions for a unique solution]
\label{thm:uniqueness}
The equation set~\eqref{eq:linear-gen}--\eqref{eq:output} has a unique nonnegative solution
under either of the following conditions:
\begin{enumerate}
    \item $\|F\|_\infty < \frac{1}{P-1}$\;\textup{(}$\|F\|_\infty < 1$ for $P=2$\textup{)}.
    \item The routing is class-independent \textup{(}$f^{(p)}_{i,j}=f^{(q)}_{i,j}$ for all
          $p,q$\textup{)} and $P\le 2$.
\end{enumerate}
\end{theorem}

Condition~(1) can be applied widely to applications with two priority classes (which we expect to be the most common implementation of \textit{DiffTry}). Under $P=2$ the condition means that $\sum_j f_{i,j}^{(p)}<1$ for all $i$, implying that a fraction of the transactions must leave in each node. In practice, it is possible that  for some nodes, say $i$, $\sum_j f_{i,j}^{(p)}=1$. However, for such cases one could slightly "tweak" the system by subtracting $\epsilon \rightarrow 0$ from one of the entries in row $i$, implying tiny departure at node $i$, which will provide a unique solution. Such "tweak" is not expected to affect system performance for practical (not ill conditioned) systems\footnote{In fact, such tiny transaction departures are naturally expected in any node due to natural transaction failures.}.
Condition~(2) captures systems in which routing does not differentiate by class, e.g.\
when DiffTry is layered onto an existing non-prioritized application.

\begin{proof}[\textbf{Outline of Proof}]
Define the map $T(\lambda):=\gamma+\varphi(\lambda)F$, where
$(\varphi(\lambda))^{(p)}_i := \lambda^{(p)}_i S^{(p)}_i(\lambda_i)$
encodes the DiffTry transformation: the output rate of class $p$ at node $i$ after
applying the retry mechanism. Any solution to~\eqref{eq:linear-gen}--\eqref{eq:output}
is a fixed point $\lambda=T(\lambda)$, 
so it remains to show $T$ has \emph{at most one} fixed point.

\textit{General uniqueness criteria.}
Since $T(\lambda)-T(\lambda')=(\varphi(\lambda)-\varphi(\lambda'))F$, then using the submultiplicativity of induced operator norms, if  $\varphi$ is $L$-Lipschitz and $\|F\|<1/L$ (in some norm), then $T$ is a
strict contraction, and the Banach Fixed-Point Theorem yields a unique fixed point.
The two conditions of Thm.~\ref{thm:uniqueness} correspond to two ways of establishing
$L\|F\|<1$.

\textit{Bounding the Lipschitz constant of $\varphi$.}
The map $\varphi$ is defined piecewise: (i) for simplicity, we assume that in the no-overload region ($\rho_i{<}1$) it
acts as the identity, so it is trivially $1$-Lipschitz. (ii) In the overload region
($\rho_i{>}1$), the per-attempt rejection probability $\bar{s}_i$ is implicitly given by $\lambda_i$ (and the constants) through eq.~\eqref{eq:s-equation}, so
$\varphi^{(p)}_i {=} \lambda^{(p)}_i S^{(p)}_i(\bar{s}_i)$ is a smooth function of
$\lambda_i$. Implicit differentiating $\varphi_i$ with respect to $\lambda_i^{(q)}$ yields a Jacobian matrix $J$ satisfying $\|J\|_\infty {\le} P{-}1$ and $\|J\|_1 {\le} 2$, since its maximum row and column sums are bounded by $P{-}1$ and $2$, respectively.
(iii) For $\lambda_i,\lambda_i'$ in different regions ($\rho_i{<}1$ and $\rho_i{>}1$), we need to show that $\|T(\lambda_i){-}T(\lambda'_i)\|{\le}L\|\lambda_i{-}\lambda'_i\|$. let $\tilde{\lambda}_i$ be the intersection of $[\lambda_i,\lambda_i']$ with $\rho_i{=}1$. The triangle inequality applied to the two resulting subsegments, each lying entirely within one region, recovers the same global bound. These per-node bounds extend directly to the global operator $\varphi$ under $\|\lambda\|_{1,w}=\sum_i w_i\|\lambda_i\|_1$ or $\|\lambda\|_{\infty,w}=\max_i w_i\|\lambda_i\|_\infty$, with the same Lipschitz constant for any choice of weights.

\textit{Bounding $\|F\|$ for Condition~(1).}
The Lipschitz bounds above give $L=P{-}1$ in $\ell_\infty$
and the hypotheses $\|F\|_\infty<\frac{1}{P-1}$ directly give the criteria for uniqueness.

\textit{Bounding $\|F\|$ for Condition~(2).} 
In principle, $\rho(F)<1$ (which holds for any open network) already guarantees
the existence of a weighted norm in which $\|F\|<1$ .
however, the resulting weights live in $\mathbb{R}^{NP}$, one per coordinate, whereas the Lipschitz bound on $\varphi$
is established \emph{blockwise},
with only $N$ weights, one per node, applied uniformly across the $P$ classes of each block.
This is where class-independence is essential: the matrix $F$ is (up to a permutation) block-diagonal with $P$
identical copies of the $N\times N$ scalar routing matrix $\hat{F}$, so
$\rho(F)=\rho(\hat{F})$. Because $\rho(\hat{F})=\rho(F)<1$, one can therefore find $N$ weights from $\hat{F}$ alone
and use them uniformly across all classes to define a valid blockwise norm in which
$\|F\|<1$.\footnote{
Choose $\beta$ with
$\rho(\hat{F}){<}\beta{<}1$. Therefore $\rho(\tfrac{1}{\beta}\hat{F}){<}1$, so the Neumann-series converges and $(I{-}\tfrac{1}{\beta}\hat{F})^{-1}{>}0$. Setting
$\tilde{w}^\top{:=}\mathbf{1}^\top(I{-}\tfrac{1}{\beta}\hat{F})^{-1}$ gives
$\tilde{w}^\top\hat{F}{=}\beta(\tilde{w}^\top{-}\mathbf{1}^\top){\le}\beta\tilde{w}^\top$. Thus, $w'_i{:=}1/\tilde{w}_i$ gives
$\|F\|_{\infty,w'}{\le}\beta{<}1$.}
Since $\varphi$ is $1$-Lipschitz
($P\le2$), uniqueness follows.
\end{proof}

\subsubsection{\textbf{Numerical verification}}
Uniqueness beyond the above conditions remains open. To examine it
empirically, we generated tens of thousands of cyclic networks using the $\mu$Bench
and DGG graph generators~\cite{mubench,DGG}, and for each we tested hundreds of
parameter configurations and thousands of distinct initializations $\lambda(0)$,
amounting to roughly $60$ million cases across $P<10$ priority classes. In no instance
did the fixed-point iteration converge to two distinct solutions, providing strong
empirical evidence that uniqueness holds universally.

\subsection{Transaction (job)  Completion Probability}
\label{sec:comp_prob}
In practice, one may want to derive the probability that a job completes its path through the network.
Let 
$S_{i}^{(p)}:=1-\bar{s}_i^{k_p+1}$ denote  the success probability of any class $p$ job entering node $i$, which can be derived according to subsection \ref{sec:solve_trans}.

In some cases, one may know for a specific job the full path it will traverse in the network. In that case the job completion probability  will be the product of its success rates over  the hops it traverses. 
Let $C^{(p)}(a_1, ..., a_n)$  denote  the successful completion probability of a job traversing nodes $a_1, ..., a_n$ and leaving the system.
So, we  get: 
$C^{(p)}(a_1, ..., a_n) = 
\prod_{i=2}^{n}S_{a_i}^{(p)}.$

In other cases, the full path might not be known in advance. In this case, let  $C_i^{(p)}$ denote  the (network) successful completion probability of a class $p$ job that starts at node $i$. This can be derived recursively from  the probability to leave the system at $i$ and from  the  completion probabilities of  the downstream neighbors of $i$. 
To this end, let $\hat{S}^{(p)}_i := 1-\sum_{j=1}^{N}f_{i,j}^{(p)}$ denote the fraction of class $p$ jobs that leave the system at node $i$. Using the definition $\hat{f}_{j, i}^{(p)}=f_{i, j}^{(p)} S_j^{(p)}$, we get:
\begin{equation}
\label{eq:completion}  
    C_i^{(p)} = \hat{S}^{(p)}_i + \sum_{j=1}^{N}f_{i,j}^{(p)}S_j^{(p)}C_j^{(p)}=\hat{S}^{(p)}_i + \sum_{j=1}^{N}\hat{f}_{j,i}^{(p)}C_j^{(p)}
\end{equation}
Now, using a vector terminology, where $C$ and $\hat{S}$ are treated as vectors of size $NP$, we get $C = \hat{S} + C\hat{F}$.
Assuming the original network is solvable, eq. \eqref{eq:completion} is solvable and given by 
\begin{equation}
C = \hat{S} + C\hat{F}=\hat{S}(I-\hat{F})^{-1}=\hat{S}\sum_{i=0}^{\infty}\hat{F}^i
\end{equation}

Finally, we derive the \textit{total completed-transactions rate}:
\begin{equation}
    \Gamma_p := \sum_{i=1}^N\gamma_i^{(p)}S_i^{(p)}C_i^{(p)} = \sum_{i=1}^N\varphi(\lambda)^{(p)}_i\hat{S}_i^{(p)}
\end{equation}
Where the last equality is because for each class $p$ the total rate of jobs arriving on the system and completing successfully must be equal to the total rate of jobs leaving the system.

\subsection{System-Wide Monotonicity of \textit{DiffTry}}
\label{sec:dominance}

\textit{DiffTry} prioritizes class $p$ by granting it more retries; in Sec. \ref{sec:evaluation}  this was shown to be rather effective for a single node. The question addressed here is whether it is effective in the context of a full network of services where the retry mechanism of a node affects (through cyclical flows) its input streams in an unknown way.  We establish this effectiveness by establishing  the monotonicity of the network converged  solution (node input rates) in the retry budgets.

\begin{definition}[$p$-dominance]
For $v,w\in\mathbb{R}^{NP}$ and a class $p$, $v$ \emph{$p$-dominates} $w$ (written
$v\ge_p w$) if $v_i^{(p)}\ge w_i^{(p)}$ and $v_i^{(q)}\le w_i^{(q)}$ for all nodes
$i$ and classes $q\ne p$. A map $f$ is \emph{dominance preserving} if for every $p$
$v\ge_p w\Rightarrow f(v)\ge_p f(w)$. 
\end{definition}
Through this subsection, a superscript $*$ marks a quantity evaluated
at a fixed point.
\begin{figure*}[t]
    \centering
    \begin{minipage}[b]{0.245\textwidth}\centering
        \includegraphics[width=\linewidth]{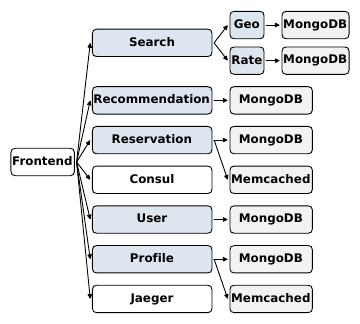    }\\[2pt]
        {\small\textbf{(a) DeathStarBench Dependency graph (from \cite{deathstarbench_github}})}
    \end{minipage}\hfill
    \begin{minipage}[b]{0.245\textwidth}\centering
        \includegraphics[width=\linewidth]{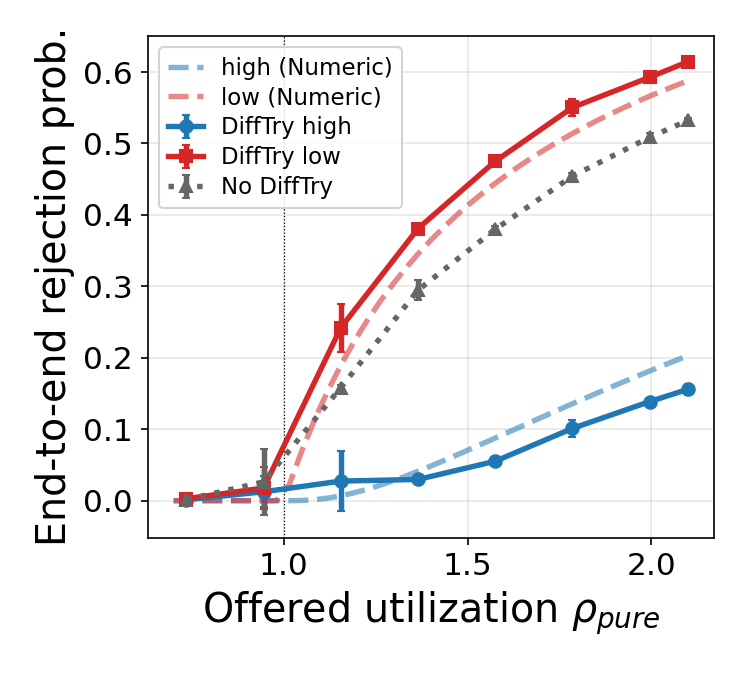}\\[2pt]
        {\small\textbf{(b) Per-class network rejection: numerical VS experiment}}
    \end{minipage}\hfill
    \begin{minipage}[b]{0.245\textwidth}\centering
        \includegraphics[width=\linewidth]{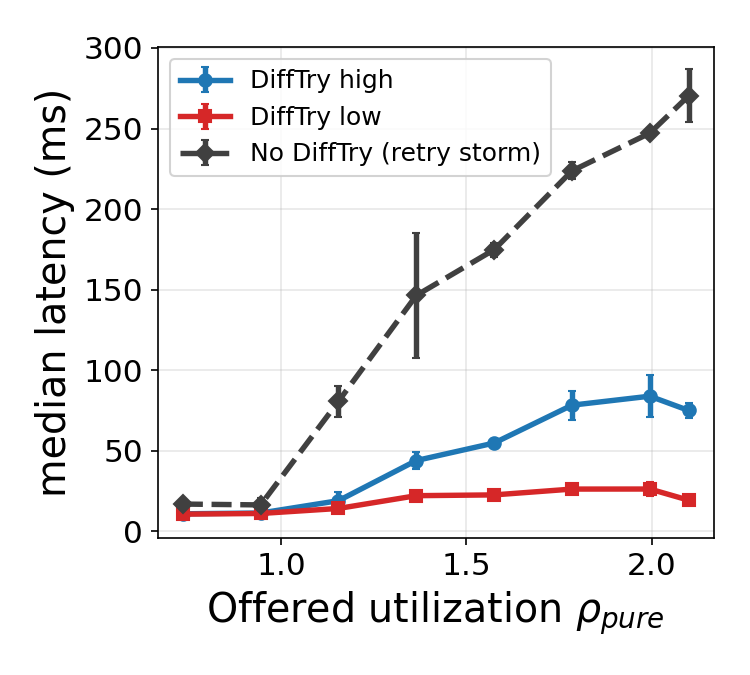}\\[2pt]
        {\small\textbf{(c) Median latency: Protection against retry Storm }}
    \end{minipage}\hfill
    \begin{minipage}[b]{0.245\textwidth}\centering
        \includegraphics[width=\linewidth]{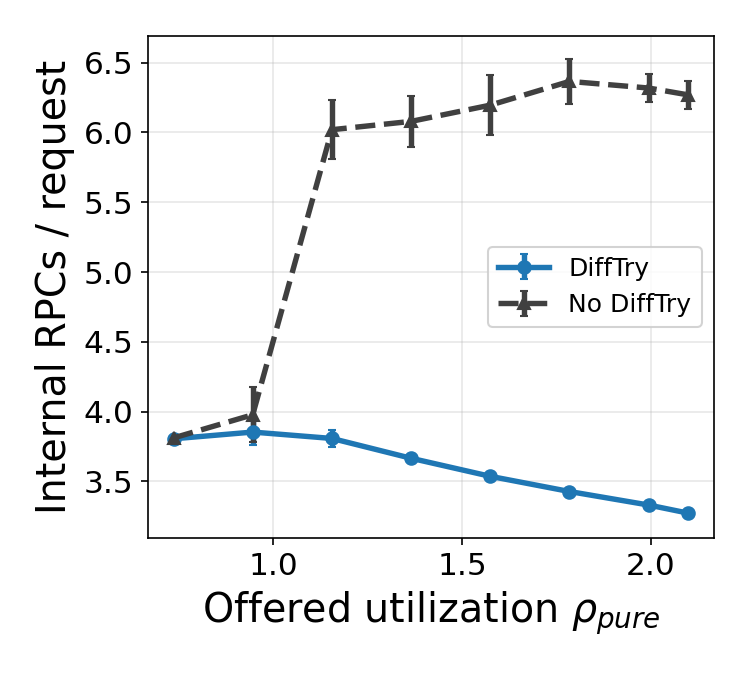}\\[2pt]
        {\small\textbf{(d) Internal RPC calls per external request}}
    \end{minipage}
    \caption{\textit{DiffTry} on live kubernetes deployment (DeathStarBench hotelReservation) - \textit{DiffTry} protects high-priority traffic, suppresses the flood's
    retry-storm latency, and improves efficiency}
    \label{fig:deathstar}
\end{figure*}
\begin{theorem}[System-wide monotonicity of fixed-points in the retry budgets]
\label{thm:dominance}
fix a class $p$ and budgets $k^{\mathrm{org}}{\le_p} k^{\mathrm{dom}}$. 
When the two systems run from any initial loads
$\lambda^{\mathrm{org}}(0){\le_p}\lambda^{\mathrm{dom}}(0)$ they converge to fixed points
satisfying:
\begin{equation}
\lambda^*_{\mathrm{dom}}\ge_p\lambda^*_{\mathrm{org}},\quad
\varphi^*_{\mathrm{dom}}\ge_p\varphi^*_{\mathrm{org}},\quad
\Gamma^*_{\mathrm{dom}}\ge_p\Gamma^*_{\mathrm{org}},
\label{eq:dom}
\end{equation}
giving consistent dominance of the prioritized system over the original, even as the load shifts between stages.
\end{theorem}
\begin{proof}[\textbf{Outline of Proof}]
We show each part preserves dominance: the node map
$\varphi$, the single iteration $T$, and its fixed point.

\emph{Single node.} 
An overloaded node shares a fixed rate of successful jobs, $\mu$, across all classes. Increasing retries (or input) allows class $p$ to claim a larger share at the expense of others; thus, $\varphi$ preserves dominance in $k$ and $\lambda$.

\emph{Single iteration.} $T$ mixes these accepted rates by routing, and since a (non-negative) mixture of $\ge_p$-ordered vectors stays
$\ge_p$-ordered, then  $T$ preserves dominance in both $k$ and $\lambda$.

\emph{Fixed points.} 
Start the two iterations from $\lambda^{\mathrm{org}}(0)\le_p\lambda^{\mathrm{dom}}(0)$, each iterative step preserves $p$-dominance (apply the single iteration step in $\lambda$, then in
$k$), so $\lambda^{\mathrm{dom}}(t)\ge_p\lambda^{\mathrm{org}}(t)$ for all $t$. Being
coordinatewise, the ordering survives to the limit, giving dominance in $\lambda^*$ and $\varphi^*$.

\emph{Completion rate.} $\Gamma^*_r=\sum_i\varphi^{*(r)}_i\hat S^{(r)}_i$ is a
nonnegative-weighted sum of the accepted rates, with weights
$\hat S^{(r)}_i=1-\sum_j f^{(r)}_{i,j}$ fixed by routing alone. The ordering on
$\varphi^*$ therefore carries through to $\Gamma^*_{\mathrm{dom}}\ge_p\Gamma^*_{\mathrm{org}}$.
\end{proof}

Starting both iterations at the same point $\lambda^*_{\mathrm{dom}}$ or $\lambda^*_{\mathrm{org}}$ gives:
\begin{corollary}
    For every fixed point
$\lambda^*_{\mathrm{org}}$ of $k^{\mathrm{org}}$ there exists a fixed point
$\lambda^*_{\mathrm{dom}}$ of $k^{\mathrm{dom}}$, satisfying \eqref{eq:dom},
and symmetrically with the roles reversed. In particular, if either system has a
unique fixed point, \eqref{eq:dom} holds for \emph{every} fixed point of the other.
\end{corollary}

\begin{remark}[Effective prioritization in networks]
Since a class-equal budget has a unique fixed point, Thm.~\ref{thm:dominance}
compares any \textit{DiffTry} equilibrium against that baseline: prioritizing class
$p$ raises its per-node successful transactions and its total completion rate
$\Gamma_p$ relative to no prioritization, even when the \textit{DiffTry} equations
admit multiple fixed points.
\end{remark}

\section{Network-Level Evaluation and Experimental Deployment in a Large Kubernetes Network}
\label{sec:exp}\label{subsec:eval_deathstar}

To examine how \textit{DiffTry} yields \emph{end-to-end} differentiation in realistic networks, we evaluate it on \emph{DeathStarBench}
\cite{deathstarbench_github, DeathStarBench}, a widely used microservice benchmark, using its
\emph{hotelReservation} application
(Fig.~\ref{fig:deathstar}(a)). We study \textit{DiffTry} in two ways: (i) \emph{Numerically}, by solving the network model of
Sec.~\ref{sec:network}, and (ii)
\emph{Experimentally}, on a live Kubernetes deployment with service mesh. 
We ask two questions: Does \textit{DiffTry}'s local, per-node mechanism compose into a \emph{network-wide} prioritization, and how well does the analysis predict a real system?

\subsection{Experiment and numerical examination Setup}
We consider two classes ($P{=}2$), high- and low-priority; external arrivals enter
at the gateway/front-end. We fix $\gamma^{(1)}{=}50$ and, to model an attacker that
floods the network (DDoS) with cheap low-priority requests, sweep
$\gamma^{(2)}\!\in\![50,250]$ (at $250$ the total offered load reaches
$\rho\!\approx\!2.1$). Service capacities $\mu_i$
are calibrated so that, without rejection, node utilization is
$0.7$.\footnote{ We also tested several utilization values, and higher maximum loads; all gave similar outcomes. Capacity is enforced by rate limiters.} We
compare
\textit{DiffTry} with $k{=}(5,1)$ against a no-prioritization baseline with symmetric
$k{=}(5,5)$. The code is available \cite{difftry_code}.

\emph{Numerically}, we solve the network fixed point (Sec.~\ref{sec:network})
for each $\gamma^{(2)}$ and report the network-level rejection $1{-}C_1^{(p)}$ of a
class-$p$ request entering the front end. 

\emph{Experimentally}, we deploy
hotelReservation on Kubernetes with the \emph{Istio}/Envoy service mesh
\cite{istio_traffic_managment}, representative of the service-mesh platforms, which are now
widely adopted in production microservice deployments \cite{cncf_servicemesh_survey}.
We realize \textit{DiffTry} entirely as a \emph{service-mesh plugin}: a small,
stateless Envoy filter that, at every sidecar, reads the request's priority (carried
in a propagated context header) and sets that hop's retry budget to the class value.
Crucially, this is a \emph{drop-in} plugin that any existing mesh can adopt with
\emph{no changes to application code}. Although the theoretical model relies on probabilistic routing, our physical experiment uses the application's native fan-out.

\subsection{Experiment and numerical examination results} 
\subsubsection{End-to-end rejection rate (Fig.~\ref{fig:deathstar}(b))} 
Without prioritization, a high-value request is rejected almost as heavily as the
attacking flood ($\approx\!54\%$ under overload). \textit{DiffTry} reduces
the high-priority rejection rate by an order of magnitude under moderate
overload ($\sim\!10\times$ at $\rho\!\approx\!1.4$), and still by $\sim\!3.5\times$ at the
heaviest flood ($\rho\!\approx\!2.1$, from $\approx\!54\%$ to $\approx\!15\%$), while the
rejection of the low-priority flood rises only slightly (to $\approx\!61\%$).
\textit{DiffTry} thus induces a clear and persistent end-to-end differentiation between
the classes.

Crucially, the \emph{measured} rejection on the live experiment tracks the
\emph{numerical} fixed-point prediction
($\approx0.03$ difference), confirming that the local per-node mechanism
manifests as a network-wide effect on a real system that the analysis predicts well. Below capacity ($\rho<1$) both classes see $\approx\!0$ rejection, so
\textit{DiffTry} is silent when the system is not overloaded, consistent with the
single-service results of Sec.~\ref{sec:evaluation}.

\subsubsection{Latency and resource consumption (Fig.~\ref{fig:deathstar}(c,d))}
Without prioritization, the uniform retries back the over-offered flood up into a
\emph{retry storm} whose median latency reaches $\approx\!270$\,ms. \textit{DiffTry} caps
the low-priority budget so the flood fails fast, collapsing its latency to
$\approx\!20$\,ms, and the high-priority latency to $\approx\!70$\,ms
($\sim\!3.6\times$ lower than baseline) -- the intrinsic cost of its $k{=}5$ retries. Thus \textit{DiffTry} \emph{drastically decreases} the latency of \emph{both}
classes far below the baseline.

\textit{DiffTry} adds negligible overhead---a single, stateless retry-budget filter---and
instead \emph{reduces} the mesh's internal work. Each external request fans out to
$\approx\!3.8$ internal gRPC calls; under overload the uniform retries inflate this to
$\approx\!6.4$ RPCs per request, whereas \textit{DiffTry} drives it \emph{below} the
fan-out floor, to $\approx\!3.3$: the flood fails fast at
the first congested service and never issues the deeper downstream calls. Per
\emph{successfully served} request the gap widens further ($\approx\!13.8$ vs
$\approx\!7.2$). 
Because rejection is enforced by rate limiters, before a request reaches the service logic, CPU and memory stay comparable.

\subsubsection{Conclusions}
These results reproduce, on a real deployment, the prioritization and protection
predicted by the analysis. Equally important, they show that \textit{DiffTry} is
trivial to adopt using a simple low-overhead service-mesh plug-in.

\section{Conclusions}\label{sec:conclusion}

We introduced \textit{DiffTry}, a distributed, stateless mechanism that embeds priority into retry logic, turning retries into a tool for differentiating request rejections under overload without explicit overload signaling or inter-service coordination.

We developed an analytical framework that lets practitioners quantify the prioritization behavior and rejection rates. We proved convergence of its solution and established conditions and evidence for uniqueness, including DAG service graphs. We further proved monotonicity of class-level success rates in the retry limits, establishing \textit{DiffTry}'s effectiveness and performance bounds even without guaranteed uniqueness.

We validated the analysis both numerically at the network level and on a live Kubernetes/Istio deployment, where \textit{DiffTry} is a drop-in service-mesh plugin whose measured per-class rejection matches the numerical prediction, requires no changes to application code, and incurs negligible overhead.

\bibliographystyle{IEEEtran}
\bibliography{References}
\end{document}